\documentclass[letterpaper, 10pt, conference]{ieeeconf}  

\usepackage{graphicx}
\usepackage{subfig}
\usepackage{amsmath}
\usepackage{amsthm} 
\usepackage{amssymb}
\usepackage[font=footnotesize]{caption} 
\usepackage{subfig} 
\usepackage[noadjust]{cite} 
\usepackage{color}
\usepackage{algorithm} 
\usepackage{algpseudocode} 
\usepackage{enumerate}
\usepackage[hidelinks,colorlinks=false]{hyperref}
\usepackage{setspace}

\usepackage{tikz}
\usetikzlibrary{calc} 
\usetikzlibrary{shapes} 
\usetikzlibrary{chains}
\usetikzlibrary{fit}
\usetikzlibrary{arrows}
\usetikzlibrary{decorations.text} 
\usetikzlibrary{decorations.markings}

\newtheorem{theorem}{Theorem}
\newtheorem{lemma}{Lemma}
\newtheorem{assumption}{Assumption}

\newtheorem{proposition}{Proposition}

\newtheorem{definition}{Definition}
\newtheorem{problem}{Problem}

\newcommand{\Null}{\mathrm{Null}}

\newcommand{\one}{\mathbf{1}}
\newcommand{\rank}{\mathrm{rank}}
\newcommand{\myspan}{\mathrm{span}}

\newcommand{\D}{\mathrm{d}}

\newcommand{\sgn}{\mathrm{sgn}}

\newcommand{\R}{\mathbb{R}}
\newcommand{\G}{\mathcal{G}}
\newcommand{\E}{\mathcal{E}}
\newcommand{\V}{\mathcal{V}}
\newcommand{\N}{\mathcal{N}}

\renewcommand{\L}{\mathcal{L}}

\newcommand{\ki}{k_{\text{\scriptsize{I}}}}
\newcommand{\kp}{k_{\text{\scriptsize{P}}}}

\graphicspath{{figures/}}

\begin{document}

\title{Bearing-Based Formation Maneuvering}
\author{Shiyu Zhao and Daniel Zelazo
\thanks{S. Zhao and D. Zelazo are with the Faculty of Aerospace Engineering, Technion - Israel Institute of Technology, Haifa, Israel.
    {\tt\small szhao@tx.technion.ac.il, dzelazo@technion.ac.il}}
}
\IEEEoverridecommandlockouts
\maketitle
\begin{abstract}
This paper studies the problem of multi-agent formation maneuver control where both of the centroid and scale of a formation are required to track given velocity references while maintaining the formation shape.
Unlike the conventional approaches where the target formation is defined by inter-neighbor relative positions or distances, we propose a bearing-based approach where the target formation is defined by inter-neighbor bearings.
Due to the invariance of the bearings, the bearing-based approach provides a natural solution to formation scale control.
We assume the dynamics of each agent as a single integrator and propose a globally stable proportional-integral formation maneuver control law.
It is shown that at least two leaders are required to collaborate in order to control the centroid and scale of the formation whereas the followers are not required to have access to any global information, such as the velocities of the leaders.
\end{abstract}
\overrideIEEEmargins

\section{Introduction}

In the conventional approaches to distributed multi-agent formation control, the target formation is defined by either inter-neighbor {relative position} or {distance} constraints.
In recent years, there has been a growing research interest in the bearing-based approach where the target formation is constrained by inter-neighbor \emph{bearings}.
The existing works on bearing-based formation control focused mainly on distributed stabilization of \emph{static} target formations using bearing-only \cite{bishop2010SCL,Eren2012IJC,Franchi2012IJRR,Cornejo2013IJRR,zhao2013SCLDistribued,Eric2014ACC,zhao2014TACBearing} or relative position measurements \cite{nima2009TR,bishopconf2011rigid,Antonio2012CDC,zhao2015ECC}.
In this paper we focus on distributed tracking control of \emph{maneuvering} target formations whose centroid and scale track given velocity references while the formation shape is maintained as desired.

In our previous work \cite{zhao2015ECC}, we proposed a linear bearing-based control law to stabilize static target formations in arbitrary dimensions.
When applied to the maneuvering case where the leaders have nonzero constant velocities, the control law proposed in \cite{zhao2015ECC} would result in constant tracking errors.
In order to eliminate the tracking errors, in this paper we adopt a proportional-integral (PI) formation control scheme, which has also been applied to solve the problem of distance-based formation maneuver control \cite{oshri2015ECC,oshri2014IACAS}.
The basic idea of the PI control scheme is to treat the constant leader velocities as an input disturbance and the impact of the disturbance can be eliminated by the integral action.
By following the PI control scheme, the stabilization control law in \cite{zhao2015ECC} can be viewed as a proportional control and an integral term of the proportional control can be added to obtain an effective formation maneuver control law.
In this paper, we show that there must exist at least two leaders and the leaders must collaborate to control the collective motion of the entire formation.

One advantage of the proposed bearing-based formation control law is that it provides a natural solution to formation scale control.
Formation scale control is a practically useful technique.
By adjusting the scale of the formation, the agents are able to dynamically respond to the environments including, for example, avoiding obstacles.
The relative position and distance based approaches have been applied to solve the formation scale control problem \cite{Coogan2012Scale,Park2014IJRNC}.
But since the inter-neighbor relative positions or distances are \emph{not} invariant to the formation scale, the two approaches lead to complicated estimation and control scheme \cite{Coogan2012Scale,Park2014IJRNC}.
As a comparison, the bearing-based approach provides a simple solution to formation scale control due to the invariance of the inter-neighbor bearings to the formation scale.
As shown in this paper, by assigning appropriate velocities to the leaders, the scale of the formation can be adjusted continuously while the formation shape can be maintained as desired.

\section{Notations and Preliminaries}\label{section_preliminary}

In this section, we present the notations and preliminaries to the bearing rigidity theory that will be used throughout the paper.

\subsection{Notations}

For any nonzero vector $x\in\R^d$ ($d\ge2$), define the orthogonal projection operator $P: \R^d\rightarrow\R^{d\times d}$ as
\begin{align*}
    P(x) \triangleq I_d - \frac{x}{\|x\|}\frac{x^T }{\|x\|}.
\end{align*}
For notational simplicity, we denote $P_x \triangleq P(x)$.
Note that $P_x$ is an orthogonal projection matrix that geometrically projects any vector onto the orthogonal compliment of $x$.
It can be easily seen that $P_x$ is positive semi-definite and satisfies $P_x^T =P_x$, $P_x^2=P_x$, and $\Null(P_x)=\myspan\{x\}$.

Denote $I_n\in\R^{n\times n}$ as the identity matrix, and $\one_n\triangleq[1,\dots,1]^T\in\R^n$.
Let $\Null(\cdot)$ be the null space of a matrix.
Let $\|\cdot\|$ be the Euclidian norm of a vector or the spectral norm of a matrix, and $\otimes$ be the Kronecker product.

\subsection{Preliminaries to Bearing Rigidity Theory}

Bearing rigidity theory provides a fundamental tool for the analysis of bearing-based formation control problems.
The results presented in the following can be found in \cite{zhao2014TACBearing}.

Consider a set of points $\{p_i\}_{i=1}^n$ in $\R^d$ ($n\ge2$ and $d\ge2$) and assume no two points are collocated.
Denote $p=[p_1^T ,\dots, p_n^T]^T \in\mathbb{R}^{dn}$.
Let $\G=(\V,\E)$ be an undirected graph that consists of a vertex set $\mathcal{V}$ and an edge set $\mathcal{E}\subseteq \mathcal{V} \times \mathcal{V}$.
Let $n=|\V|$ and $m=|\E|$.
The set of neighbors of vertex $i$ is denoted as $\mathcal{N}_i\triangleq\{j \in \mathcal{V}: \ (i,j)\in \mathcal{E}\}$.
A \emph{formation}, denoted as $\G(p)$, is a graph $\G$ with vertex $i$ in the graph mapped to the point $p_i$ for all $i\in\V$.

For a formation $\G(p)$, define the \emph{edge vector} and the \emph{bearing}, respectively, as
\begin{align*}
e_{ij}\triangleq p_j-p_i, \quad g_{ij}\triangleq e_{ij}/\|e_{ij}\|, \quad \forall(i,j)\in\E.
\end{align*}
The bearing $g_{ij}$ is a unit vector.
Note $e_{ij}=-e_{ji}$ and $g_{ij}=-g_{ji}$.
Assign a direction to each edge in $\G$ and express the edge vector and the bearing for the $k$th directed edge in the oriented graph as $e_{k}\triangleq p_j-p_i$, $g_{k}\triangleq {e_{k}}/{\|e_{k}\|}$, $\forall k\in\{1,\dots,m\}$, respectively.
Define the \emph{bearing function} $F_B: \R^{dn}\rightarrow\R^{dm}$ as
$$F_B(p)\triangleq [g_1^T ,\dots, g_m^T]^T.$$
The bearing function describes all the bearings in the formation.
The \emph{bearing rigidity matrix} is defined as the Jacobian of the bearing function,
\begin{align*}
    R_B(p) \triangleq \frac{\partial F_B(p)}{\partial p}\in\R^{dm\times dn}.
\end{align*}
Let $\delta p$ be a variation of $p$.
If $R_B(p)\delta p=0$, then $\delta p$ is called an \emph{infinitesimal bearing motion} of $\G(p)$.
A formation always has two kinds of \emph{trivial} infinitesimal bearing motions: translation and scaling of the entire formation.
The next definition is one of the most important concepts in bearing rigidity theory.

\begin{definition}[{Infinitesimal Bearing Rigidity}]\label{definition_infinitesimalParallelRigid}
    A formation is \emph{infinitesimally bearing rigid} if all the infinitesimal bearing motions of the formation are trivial.
\end{definition}

The following are necessary and sufficient conditions for infinitesimal bearing rigidity.

\begin{theorem}[\cite{zhao2014TACBearing}]\label{theorem_IBR_NSCondition}
    For any formation $\G(p)$, the following statements are equivalent:
    \begin{enumerate}[(a)]
    \item $\G(p)$ is {infinitesimally bearing rigid};
    \item $\G(p)$ can be uniquely determined up to a translation and a scaling factor by the inter-neighbor bearings;
    \item $\rank(R_B)=dn-d-1$;
    \item $\Null(R_B)=\myspan\{\one_n\otimes I_d, p\}$.
    \end{enumerate}
\end{theorem}

\section{Problem Statement of Bearing-Based Formation Maneuvering}
Consider a formation of $n$ agents in $\R^d$ ($n\ge2$, $d\ge2$).
Let $p_i\in\R^d$ and $v_i \in \R^d$ be the position and velocity of agent $i \in \V=\{1,\dots,n\}$.
Suppose the velocity of the first $n_\ell$ agents are given.
Then the first $n_\ell$ agents are called \emph{leaders} and the rest $n_f$ agents are called \emph{followers}.
Note $0\le n_\ell \le n$ and $n_\ell+n_f=n$.
Let $\V_\ell=\{1,\dots,n_\ell\}$ and $\V_f=\{n_\ell+1,\dots,n\}$ be the index sets of the leaders and followers, respectively.
The dynamics of the leaders and followers are given by
\begin{align*}
\dot{p}_i(t)&=v^*_i(t), \quad  i\in\V_\ell,\\
\dot{p}_i(t)&=v_i(t), \quad i\in\V_f,
\end{align*}
where $\{v_i^*(t)\}_{i\in\V_\ell}$ are the velocity references assigned to the leaders, and $\{v_i(t)\}_{i\in\V_f}$ are the control inputs to be designed for the followers.
Denote $p_\ell=[p_1^T ,\dots, p_{n_\ell}^T ]^T $, $p_f=[p_{n_\ell+1}^T  ,\dots, p_{n}^T ]^T $, and $p=[p_\ell^T, p_f^T]^T $.
Furthermore, denote $v_\ell^*=[(v_1^*)^T  ,\dots, (v_{n_\ell}^*)^T ]^T $, $v_f=[v_{n_\ell+1}^T  ,\dots, v_n^T]^T $, and $v=[(v_\ell^*)^T , v_f^T]^T $.

Suppose the underlying information flow among the agents can be described by a fixed and undirected graph $\G=(\V,\E)$.
If $(i,j)\in\E$, then agents $i$ and $j$ can obtain the relative positions of each other.
Suppose a desired formation shape is defined by the constant inter-neighbor bearing constraints $\{g_{ij}^*\}_{(i,j)\in\E}$.
The problem of bearing-based formation maneuvering is stated as below.

\begin{problem}[{Bearing-Based Formation Maneuvering}]\label{problem_bearingFormationManeuver}
Consider a formation $\G(p)$ where the velocities of the leaders are assigned as $\{v_i^*(t)\}_{i\in\V_\ell}$.  Design the control input for each follower $v_i(t)$ ($i\in\V_f$) based on the relative position measurements $\{p_i(t)-p_j(t)\}_{j\in\N_i}$ such that the inter-neighbor bearings converge to the desired values, i.e., $g_{ij}(t)\rightarrow g_{ij}^*$ for all $(i,j)\in\E$ as $t\rightarrow \infty$.
\end{problem}

When the velocity of each leader is zero, Problem~\ref{problem_bearingFormationManeuver} would become the bearing-based formation stabilization problem studied in \cite{zhao2015ECC}.

In Problem~\ref{problem_bearingFormationManeuver} the formation is required to converge to a target formation that is jointly determined by the inter-neighbor bearings and the leaders.
This target formation is formally defined below.
\begin{definition}[Target Formation]\label{definition_targetFormation}
Let $\G(p^*(t))$ be the \emph{target formation} that satisfies the following constraints for all $t>0$:
\begin{enumerate}[(a)]
\item Bearing: $(p_j^*(t)-p_i^*(t))/\|p_j^*(t)-p_i^*(t)\|=g_{ij}^*, \forall (i,j)\in\E$,
\item Leader: $p_i^*(t)=p_i(t), \forall i\in\V_\ell$.
\end{enumerate}
\end{definition}

Since the leaders have nonzero velocities, the centroid and scale of the target formation are \emph{time-varying}.
The shape of the target formation is, however, \emph{fixed} since the bearing constraints are constant.
Based on the notion of the target formation $\G(p^*(t))$, Problem~\ref{problem_bearingFormationManeuver} can be equivalently stated as designing a control law such that the formation $p(t)$ can converge to the target formation $p^*(t)$, i.e., $p_i(t)\rightarrow p_i^*(t)$ for all $i\in\V_f$ as $t\rightarrow \infty$.

One key problem that follows the definition of the target formation is whether the target formation exists and is unique.
In fact, this problem is equivalent to the localizability problem in bearing-only network localization \cite{zhao2015NetLocalization}.
As shown in \cite{zhao2015NetLocalization}, a formation can be uniquely determined by the bearings and the leaders if and only if the formation is \emph{localizable}.
A variety of conditions for localizability have been proposed in \cite{zhao2015NetLocalization}.
One useful sufficient condition is that a formation would be localizable if the formation is infinitesimally bearing rigid and has at least two leaders.
This sufficient condition is intuitively easy to understand.
Specifically, if the formation is infinitesimally bearing rigid, then it can be uniquely determined up to a translation and a scaling factor by the bearings.
The translation and scaling ambiguity can be further eliminated by the introduction of at least two leaders.
Then the formation can be fully and uniquely determined.
Due to space limitations, we omit the details and simply make the following assumption.

\begin{assumption}\label{assumption_targetFormationUniqueExist}
The target formation $\G(p^*(t))$ is infinitesimally bearing rigid and has at least two leaders.
\end{assumption}

The mathematical condition implied by Assumption~\ref{assumption_targetFormationUniqueExist} will be given later (see Lemma~\ref{lemma_LffNonsingular}).
It is worth noting that the infinitesimal bearing rigidity is merely sufficient but not necessary to ensure the existence and uniqueness of the target formation.
For both necessary and sufficient conditions, please see \cite{zhao2015NetLocalization}.

\section{Proposed Control Law \\and Convergence Analysis}

In this section we first propose a distributed PI control law to solve Problem~\ref{problem_bearingFormationManeuver} and then prove the global formation stability under the control law.

\subsection{A Distributed PI Control Law}

For each follower, the proposed control law is
\begin{align}\label{eq_controlLaw_maneuver_element_PI}
    \dot{p}_i(t)
    &=-\kp \underbrace{\sum_{j\in\N_i} P_{g_{ij}^*} (p_i(t)-p_j(t))}_{\text{proportional}} \nonumber\\
    &\quad -\ki \underbrace{\int_0^t\sum_{j\in\N_i} P_{g_{ij}^*} (p_i(\tau)-p_j(\tau))\D \tau}_{\text{integral}}, \quad i\in\V_f,
\end{align}
where $P_{g_{ij}^*}=I_d-g_{ij}^*(g_{ij}^*)^T $, and $\kp$ and $\ki$ are positive constant control gains.
Control law \eqref{eq_controlLaw_maneuver_element_PI} is distributed because the control of each agent only requires the relative positions of its neighbors.
The control law consists of a proportional term and an integral term.
When $\ki=0$, the control law would become the one proposed in \cite{zhao2015ECC}.
If the velocities of the leaders are zero, the proportional control alone is able to stabilize the target formation.
But if the velocities of the leaders are nonzero, the integral control is required to eliminate the steady state tracking error.

By defining a new state for the integral term, control law \eqref{eq_controlLaw_maneuver_element_PI} can be rewritten as
\begin{align}\label{eq_controlLaw_maneuver_element}
    \dot{p}_i(t)&=-\kp \sum_{j\in\N_i} P_{g_{ij}^*} (p_i(t)-p_j(t)) -\ki \xi_i(t), \nonumber \\
    \dot{\xi}_i(t)&= \sum_{j\in\N_i} P_{g_{ij}^*} (p_i(t)-p_j(t)), \quad i\in\V_f.
\end{align}

We next derive the matrix expression of control law \eqref{eq_controlLaw_maneuver_element}, which will be useful for the convergence analysis.
Let $\L\in\R^{dn\times dn}$ be a matrix with the $ij$th block submatrix as
\begin{align*}
\left\{
  \begin{array}{ll}
      [\L]_{ij}=0_{d\times d}, & i\ne j, (i,j)\notin\E, \\
      {[\L]_{ij}}=-P_{g_{ij}^*}, & i\ne j, (i,j)\in\E, \\ 
      {[\L]_{ii}}=\sum_{k\in\N_i}P_{g_{ik}^*}, & i=j, i\in\V. \\
  \end{array}
\right.
\end{align*}
The matrix $\L$ can be interpreted as a matrix-weighted graph Laplacian.
We call $\L$ the \emph{bearing Laplacian} since it carries the information of both the underlying graph and the bearings of the formation.
Partition $\L$ into the following form
\begin{align*}
    \L=\left[
         \begin{array}{cc}
           \L_{\ell\ell} & \L_{\ell f} \\
           \L_{f\ell} & \L_{ff} \\
         \end{array}
       \right],
\end{align*}
where $\L_{\ell\ell}\in\R^{dn_\ell\times dn_\ell}$, $\L_{\ell f}=\L_{f\ell}^T \in\R^{dn_\ell\times d n_f}$, and $\L_{ff}\in\R^{dn_f \times dn_f}$.
Then, it is straightforward to see that the matrix expression of control law \eqref{eq_controlLaw_maneuver_element} is
\begin{align}\label{eq_controlLaw_maneuver_matrix}
    \dot{p}_f(t)&=-\kp\left(\L_{ff}p_f(t)+\L_{f\ell}p_\ell(t)\right)-\ki \xi(t), \nonumber\\
    \dot{\xi}(t)&=\L_{ff}p_f(t)+\L_{f\ell}p_\ell(t),
\end{align}
where $\xi(t)=[\xi_{n_\ell+1}(t)^T  ,\dots, \, \xi_n(t)^T ]^T \in\R^{dn_f}$.

\subsection{Convergence Analysis}

Define the tracking error for the followers as
\begin{align}\label{eq_error_follower}
\delta(t)\triangleq p_f(t)-p_f^*(t),
\end{align}
where $p_f(t)$ is the real position of the followers and $p_f^*(t)$ is the \emph{time-varying} expected position of the followers in the target formation.
The aim of the convergence analysis is to show that $\delta(t)$ converges to zero.
To that end, we need to first calculate $p_f^*(t)$.
The following two results on the bearing Laplacian are useful.

\begin{lemma}\label{lemma_LP=0}
Any formation $\G(p)$ that satisfies the bearing constraints $\{g_{ij}^*\}_{(i,j)\in\E}$ satisfies $\L p=0$ and $\L_{ff}p_f+\L_{f\ell}p_l=0$.
\end{lemma}
\begin{proof}
Note
\begin{align*}
\L p&=
\left[
         \begin{array}{cc}
           \L_{\ell\ell} & \L_{\ell f} \\
           \L_{f\ell} & \L_{ff} \\
         \end{array}
       \right]
       \left[
         \begin{array}{c}
           p_\ell \\
           p_f \\
         \end{array}
       \right]
=
       \left[
         \begin{array}{c}
           \vdots \\
           \sum_{j\in\N_i} P_{g_{ij}^*} (p_i-p_j)\\
           \vdots \\
         \end{array}
       \right].
\end{align*}
By the elementwise expression of $\L p$, it is obvious that $\L p=0$ if $\G(p)$ satisfies the bearing constraints, i.e, $(p_j-p_i)/\|p_j-p_i\|=g_{ij}^*$.
By the partitioned block matrix expression of $\L p$, it can be seen that $\L p=0$ implies $\L_{ff}p_f+\L_{f\ell}p_l=0$.
\end{proof}
\begin{lemma}[\cite{zhao2015NetLocalization}]\label{lemma_LffNonsingular}
The matrix $\L_{ff}$ is positive definite if Assumption~\ref{assumption_targetFormationUniqueExist} holds.
\end{lemma}

Since the target formation $p^*(t)$ satisfies the bearing constraints, it follows from Lemma~\ref{lemma_LP=0} that $\L_{ff}p_f^*(t)+\L_{f\ell}p_\ell(t)=0$.
Furthermore, since $\L_{ff}$ is positive definite by Lemma~\ref{lemma_LffNonsingular}, we have
\begin{align}\label{eq_expectedFollowerPosition}
{p}_f^*(t)=-\L_{ff}^{-1}\L_{f\ell}p_\ell(t).
\end{align}
Substituting \eqref{eq_error_follower} and \eqref{eq_expectedFollowerPosition} into \eqref{eq_controlLaw_maneuver_matrix} yields
\begin{align}\label{eq_delta_dynamics}
\dot{\delta}(t)&=-\kp\L_{ff}\delta(t) -\ki \xi(t) + \L_{ff}^{-1}\L_{f\ell}v_\ell^*(t), \nonumber \\
\dot{\xi}(t)&=\L_{ff}\delta(t).
\end{align}
The $\delta$- and $\xi$-dynamics given above can be written in a compact form as
\begin{align}\label{eq_deltaxi_dynamics}
\left[
  \begin{array}{c}
    \dot{\delta} \\
    \dot{\xi} \\
  \end{array}
\right]
=\underbrace{\left[
   \begin{array}{cc}
     -\kp \L_{ff} & -\ki I_{n_f} \\
     \L_{ff} & 0 \\
   \end{array}
 \right]}_{A}
 \left[
  \begin{array}{c}
    {\delta} \\
    {\xi} \\
  \end{array}
\right]
+\left[
  \begin{array}{c}
    \L_{ff}^{-1}\L_{f\ell} \\
    0 \\
  \end{array}
\right]v_\ell^*(t).
\end{align}

\begin{lemma}
The state matrix $A$ in \eqref{eq_deltaxi_dynamics} is Hurwitz for any $\kp, \ki>0$.
\end{lemma}
\begin{proof}
Suppose $\lambda$ is an eigenvalue of $A$. Then,
\begin{align*}
\det(\lambda I-A)
&=\det\left(\left[
   \begin{array}{cc}
     \lambda I+\kp \L_{ff} & \ki I \\
     -\L_{ff} & \lambda I \\
   \end{array}
 \right]\right) \\
&=\det\left(\lambda^2 I+\kp\lambda\L_{ff}+\ki \L_{ff}\right) \\
&=\det\left((\kp \lambda +\ki)\left(\frac{\lambda^2I}{\kp\lambda+\ki}+\L_{ff}\right)\right).
\end{align*}
As a result, $\det(\lambda I-A)=0$ implies either $\lambda=-\ki/\kp<0$ or
\begin{align*}
\frac{\lambda^2}{\kp\lambda+\ki}=-\sigma,
\end{align*}
where $\sigma$ is the eigenvalue of $\L_{ff}$.
Since $\L_{ff}$ is symmetric positive definite, we know $\sigma$ is positive and real.
The solution to the above equation can be easily calculated and it can be shown that $\lambda<-\ki/\kp<0$.
\end{proof}

Since $A$ is Hurwitz, system \eqref{eq_deltaxi_dynamics} is stable.
When $v_\ell^*(t)$ is time-varying, the followers are not able to perfectly track the leaders.
The fundamental reason is that the followers do not have access to the time-varying velocities of the leaders.
When $v_\ell^*(t)$ is constant, the tracking error will globally converge to zero due to the integral action.

\begin{theorem}[Global Convergence]\label{theorem_mainConvergenceResult}
When $v_\ell^*(t) = {\bf v}_{\ell}^*$ is constant, $\delta(t)$ and $\xi(t)$ exponentially and globally converges to $\delta(\infty)=0$ and $\xi(\infty)=\L_{ff}^{-1}\L_{f\ell}{\bf v}_{\ell}^*/\ki$, respectively.
As a result,
\begin{align*}
p_f(t)&\rightarrow p_f^*(t)=-\L_{ff}^{-1}\L_{f\ell}p_\ell(t), \\
\dot{p}_f(t)&\rightarrow v_f^*=-\L_{ff}^{-1}\L_{f\ell}{\bf v}_{\ell}^*,
\end{align*}
as $t\rightarrow \infty$.
\end{theorem}
\begin{proof}
Since $A$ is Hurwitz and ${\bf v}_{\ell}^*$ is constant, $\delta(t)$ and $\xi(t)$ globally and exponentially converge to the steady state values of $\delta$ and $\xi$.
By letting $\dot{\delta}=0$ and $\dot{\xi}=0$, we can easily verify that the steady state values are $\delta(\infty)=0$ and $\xi(\infty)=\L_{ff}^{-1}\L_{f\ell}{\bf v}_{\ell}^*/\ki$.
It follows from $\delta(t)\rightarrow0$ and $\dot{\delta}(t)\rightarrow0$ that $p_f(t)\rightarrow p_f^*(t)$ and $\dot{p}_f(t)\rightarrow \dot{p}_f^*(t)=v_f^*$.
\end{proof}

By substituting $\xi(\infty)=\L_{ff}^{-1}\L_{f\ell}{\bf v}_{\ell}^*/\ki$ back into \eqref{eq_delta_dynamics}, it can be seen that the integral term $\xi$ finally eliminates the impact of the ``disturbance'' ${\bf v}_{\ell}^*$.

\subsection{Centroid and Scale Dynamics}

We next study how the leaders should move in order to realize the desired translational and scaling formation maneuvers under the action of the proposed control law.
In order to do that, we need to first define the centroid and scale of a formation and then analyze their dynamics.
Define the centroid $c(p^*(t))$ and the scale $s(p^*(t))$ of $p^*(t)$ as
\begin{align*}
c(p^*(t))&\triangleq\frac{1}{n}\sum_{i\in\V}p_i^*(t)=\frac{1}{n}(\one_n\otimes I_d)^T {p}^*(t),\\
s(p^*(t))&\triangleq\sqrt{\frac{1}{n}\sum_{i\in\V}\|p_i^*(t)-c(p^*)\|^2}\\
&=\frac{1}{\sqrt{n}}\|p^*(t)-\one_n\otimes c(p^*)\|.
\end{align*}

\newcommand{\vc}{v_{\mathrm{c}}}

\begin{proposition}[Translational Maneuvering]\label{proposition_translationManeuver}
If the velocity of each leader is constant and satisfies
$$v_i^*={\bf v}_c, \quad i\in\V_\ell,$$
where ${\bf v}_c\in\R^d$ is a common velocity, then
\begin{align*}
\dot{c}(p^*(t))\equiv{\bf v}_c, \quad \dot{s}(p^*(t))\equiv0,
\end{align*}
which means the target formation $p^*(t)$ moves at the common velocity while the scale is fixed.
\end{proposition}
\begin{proof}
Since $v_i^*(t)={\bf v}_c, \forall i\in\V_\ell$, we have ${v}_{\ell}^*(t)=\one_{n_\ell}\otimes {\bf v}_c$.
It follows from $\L(\one_n\otimes {\bf v}_c)=0$ that $\L_{f\ell}(\one_{n_\ell}\otimes {\bf v}_c)+\L_{ff}(\one_{n_f}\otimes {\bf v}_c)=0$.
As a result, $$v_f^*=-\L_{ff}^{-1}\L_{f\ell}{v}_{\ell}^*=\one_{n_f}\otimes {\bf v}_c.$$
Consequently, $\dot{p}^*=[({v}_{\ell}^*)^T ,(v_f^*)^T ]^T =\one_n\otimes {\bf v}_c$.
Substituting $\dot{p}^*$ into $\dot{c}(p^*)$ and $\dot{s}(p^*)$ gives
\begin{align*}
\dot{c}(p^*)&=\frac{1}{n}(\one_n\otimes I_d)^T \dot{p}^*=\frac{1}{n}(\one_n\otimes I_d)^T (\one_n\otimes {\bf v}_c)={\bf v}_c, \\
\dot{s}(p^*)&=\frac{1}{\sqrt{n}}\frac{(p-\one_n\otimes c(p^*))^T }{\|p-\one_n\otimes c(p^*)\|}\dot{p}^*\\
&=\frac{1}{\sqrt{n}}\frac{(p-\one_n\otimes c(p^*))^T }{\|p-\one_n\otimes c(p^*)\|}(\one_n\otimes {\bf v}_c)=0.
\end{align*}
\end{proof}

\begin{proposition}[Scaling Maneuvering]\label{propostion_scaleManeuvering}
If the velocity of each leader is constant and satisfies
\begin{align}\label{eq_velocityForScale}
v_i^*=\alpha_i \frac{p_i^*(t)-c(p^*)}{\|p_i^*(t)-c(p^*)\|},\quad i\in\V_\ell,
\end{align}
where $\alpha_i\in\R$ is constant and satisfies $\alpha_i/\|p_i^*(t)-c(p^*)\|=\alpha_j/\|p_j^*(t)-c(p^*)\|$ for all $i,j\in\V_\ell$, then
\begin{align*}
\dot{c}(p^*(t))\equiv0, \quad \dot{s}(p^*(t))\equiv \sgn(\alpha_i)\sqrt{\frac{1}{n}\sum_{i\in\V}\alpha_i^2},
\end{align*}
which means the scale of the target formation $p^*(t)$ is continuously varying while the centroid is fixed.
\end{proposition}
\begin{proof}
Since $\alpha_i/\|p_i^*(t)-c(p^*)\|=\alpha_j/\|p_j^*(t)-c(p^*)\|$ for all $i,j\in\V_\ell$, there exists $k(t)$ such that
$$\alpha_i=k(t)\|p_i^*(t)-c(p^*)\|$$ and hence $v_i^*=k(t)(p_i^*(t)-c(p^*))$ for all $i\in\V_\ell$.
Then,  $v_{\ell}^*$ can be expressed as $v_{\ell}^*=k(t) (p^*-\one_{n_\ell}\otimes c(p^*))$.
It follows from $\L(p^*-\one_n\otimes c(p^*))=0$ that $\L_{f\ell}(p_\ell^*-\one_{n_\ell}\otimes c(p^*))+\L_{ff}(p_f^*-\one_{n_f}\otimes c(p^*))=0$.
As a result, $$v_f^*=-\L_{ff}^{-1}\L_{f\ell}v_{\ell}^*=k(t)(p_f^*-\one_{n_f}\otimes c(p^*)).$$
Consequently, $\dot{p}^*=[(v_{\ell}^*)^T ,(v_f^*)^T ]^T =k(t)(p^*-\one_{n}\otimes c(p^*))$.
Substituting $\dot{p}^*$ into $\dot{c}(p^*)$ and $\dot{s}(p^*)$ gives
\begin{align*}
\dot{c}(p^*)&=\frac{1}{n}(\one_n\otimes I_d)^T \dot{p}^*\\
&=\frac{1}{n}(\one_n\otimes I_d)^T (p^*-\one_{n}\otimes c(p^*))k(t)=0,
\end{align*}
and
\begin{align*}
\dot{s}(p^*)&=\frac{1}{\sqrt{n}}\frac{(p^*-\one_n\otimes c(p^*))^T }{\|p^*-\one_n\otimes
c(p^*)\|}\dot{p}^*\\
&=k(t)\frac{1}{\sqrt{n}}\frac{(p^*-\one_n\otimes c(p^*))^T }{\|p^*-\one_n\otimes c(p^*)\|}(p^*-\one_{n}\otimes c(p^*)) \\
&=k(t)\frac{1}{\sqrt{n}}\|p^*-\one_n\otimes c(p^*)\|\\
&=k(t)\sqrt{\frac{1}{n}\sum_{i\in\V}\|p_i^*(t)-c(p^*)\|^2}\\
&=k(t)\sqrt{\frac{1}{n}\sum_{i\in\V}\frac{\alpha_i^2}{k^2(t)}}
=\sgn(\alpha_i)\sqrt{\frac{1}{n}\sum_{i\in\V}\alpha_i^2}.
\end{align*}
\end{proof}

In Proposition~\ref{propostion_scaleManeuvering}, if $\alpha_i>0$, the velocity of each agent is pointing from the fixed centroid to the agent and hence the formation scale dilates; otherwise, if $\alpha_i<0$ the formation scale shrinks.
It should be noted that $v_i^*$ given in \eqref{eq_velocityForScale} is constant though $p_i^*(t)$ is time-varying.

When the velocity of the leaders is a linear combination of a translational and a scaling term, both the centroid and the scale of the formation will be time-varying.
Denote $c(p(t))$ and $s(p(t))$ as the centroid and the scale of the real formation $p(t)$, respectively.
By combining Theorem~\ref{theorem_mainConvergenceResult}, Propositions~\ref{proposition_translationManeuver}, and \ref{propostion_scaleManeuvering}, we obtain the following theorem.

\begin{theorem}\label{theorem_bothTransAndScale}
Under Assumption~\ref{assumption_targetFormationUniqueExist}, if the velocity of each leader is constant and there exist constant ${\bf v}_c\in\R^d$ and $\alpha_i\in\R$ such that $v_i^*$ can be decomposed to
\begin{align}\label{eq_vistart_transScale}
v_i^*={\bf v}_c + \alpha_i \frac{p_i^*(t)-c(p^*)}{\|p_i^*(t)-c(p^*)\|}, \quad \forall i\in\V_\ell,
\end{align}
where $\alpha_i$ satisfies $\alpha_i/\|p_i^*(t)-c(p^*)\|=\alpha_j/\|p_j^*(t)-c(p^*)\|$ for all $i,j\in\V_\ell$.
As a result,
\begin{align*}
\dot{c}(p^*(t))\equiv{\bf v}_c, \quad \dot{s}(p^*(t))\equiv\sgn(\alpha_i)\sqrt{\frac{1}{n}\sum_{i\in\V}\alpha_i^2}.
\end{align*}
Furthermore, $\dot{c}(p(t))$ and $\dot{s}(p(t))$ globally converge to $\dot{c}(p^*(t))$ and $\dot{s}(p^*(t))$, respectively.
\end{theorem}
\begin{proof}
Since the target formation $p^*(t)$ satisfies the bearing constraints, it follows from Lemma~\ref{lemma_LP=0} that $\L p^*(t)=0$, which implies $\L v^*(t)=0$.
Since the bearing constraints imply infinitesimal bearing rigidity, then $\Null(\L)=\myspan\{\one_{n}\otimes I_d, p^*\}$.
As a result, if $v^*$ is constant, it can always be expressed as a linear combination of $\one_n\otimes I_d$ and $p^*-\one_n\otimes c(p^*)$.
Therefore, equation \eqref{eq_vistart_transScale} holds.
The rest of the theorem directly follows from Theorem~\ref{theorem_mainConvergenceResult}, Propositions~\ref{proposition_translationManeuver}, and \ref{propostion_scaleManeuvering}.
\end{proof}

Two remarks on Theorem~\ref{theorem_bothTransAndScale} are given below.
Firstly, in order to achieve the desired translational or scaling maneuvering, the leaders must collaborate and share some global information like a common velocity and the centroid of the formation.
Secondly, $v^*_i$ given in \eqref{eq_vistart_transScale} is constant because the unit vector $({p_i^*(t)-c(p^*)})/{\|p_i^*(t)-c(p^*)\|}$ is invariant to translational and scaling maneuvers though $p^*_i(t)$ is not.

\section{Simulation Examples}

Figure~\ref{fig_sim_2DObstacle} and Figure~\ref{fig_sim_3DObstacle} give two comprehensive examples to demonstrate that the proposed control law can control the formation scale and translation to perform sophisticated maneuvers such as moving the formation through a narrow passage while maintaining the desired formation shape.
In each example, the target formation is infinitesimally bearing rigid and has two leaders.
The leader velocities are piecewise constant.
In the two examples, the leaders have a common forward velocity such that the entire formation always moves forward.
In order to adjust the formation scale, the leaders will have an additional velocity to move toward or away from each other to adjust the distance between them.
It is also worthwhile mentioning that although the tracking errors converge to zero only when the leader velocities are constant, the proposed control law can still give satisfactory performance in the cases where the leader velocities are piecewise continuous.

\begin{figure}
  \centering
  \subfloat[Formation trajectory (the leaders are marked by red and blue circles, and the followers by green circles).]{\includegraphics[width=0.9\linewidth]{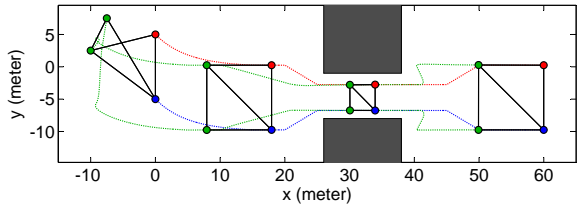}}\\
  \subfloat[Leader velocity commands (both leaders, red and blue, have the same $x$-component input).]{\includegraphics[width=0.9\linewidth]{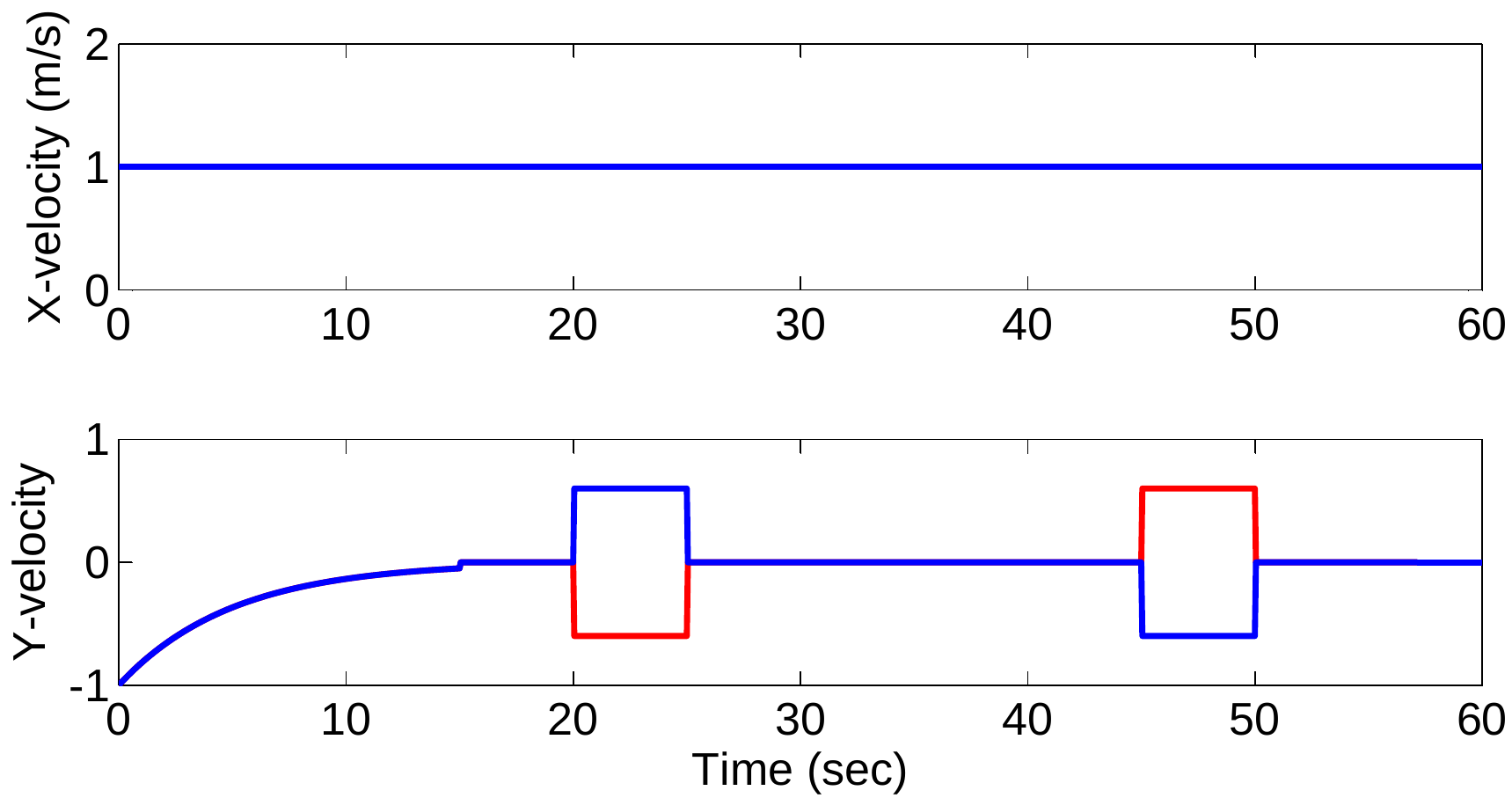}}\\
  \subfloat[Bearing error: $\sum_{(i,j)\in\E}\|g_{ij}(t)-g_{ij}^*\|$]{\includegraphics[width=0.9\linewidth]{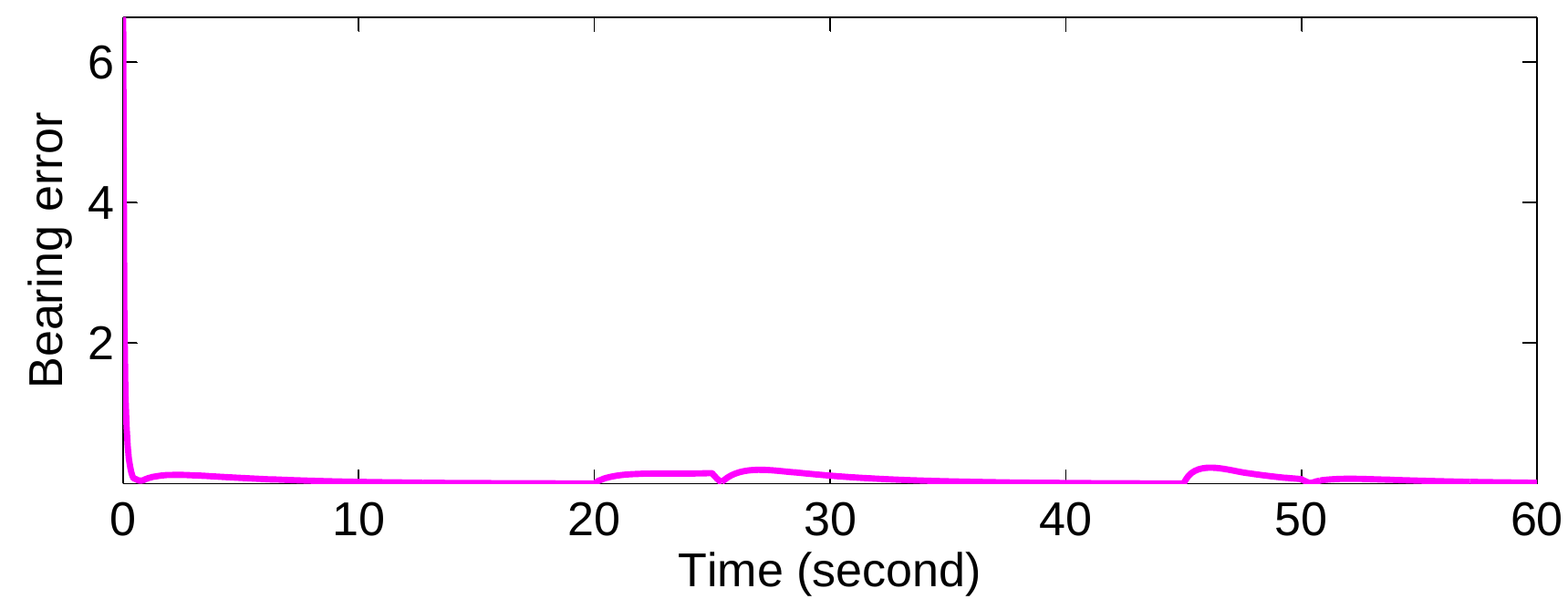}}
  \caption{A 2D example to demonstrate the formation maneuver through a narrow passage in the plane.
  }
  \label{fig_sim_2DObstacle}
\end{figure}

\begin{figure}
  \centering
  \subfloat[Formation trajectory (the leaders are marked by red and blue circles, and the followers by green circles).]{\includegraphics[width=0.9\linewidth]{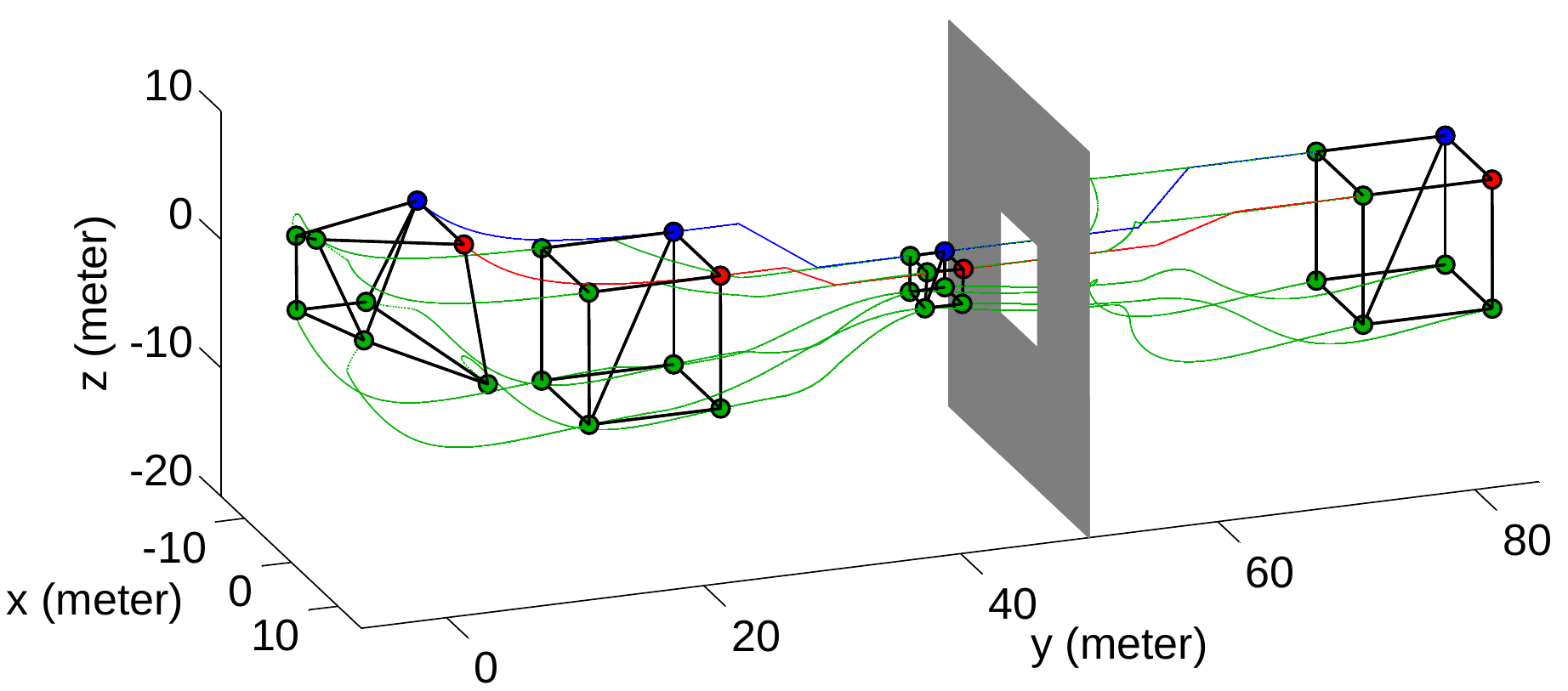}}\\
  \subfloat[Leader velocity commands (both leaders, red and blue, have the same y- and z-component inputs).]{\includegraphics[width=0.9\linewidth]{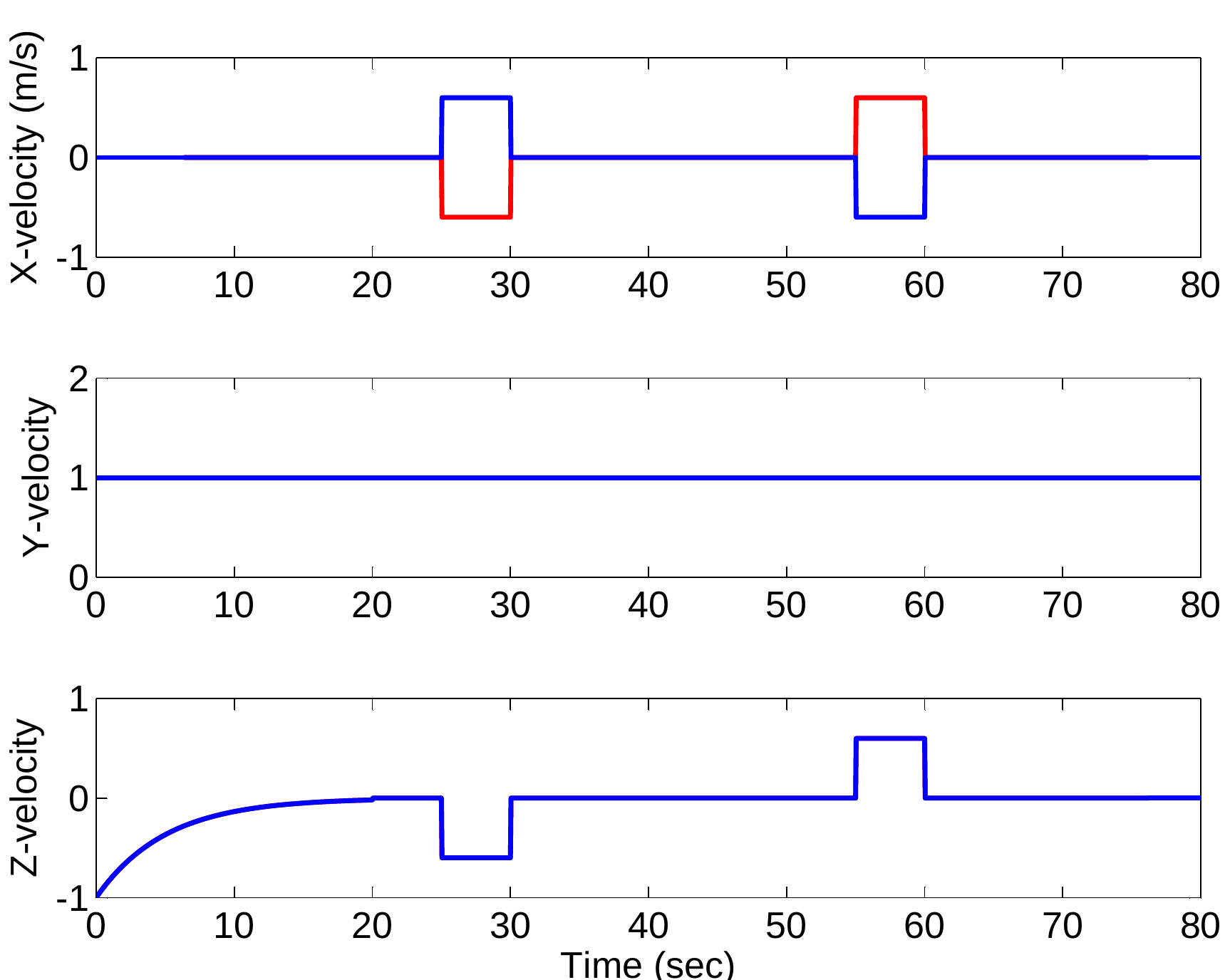}}\\
  \subfloat[Bearing error: $\sum_{(i,j)\in\E}\|g_{ij}(t)-g_{ij}^*\|$]{\includegraphics[width=0.9\linewidth]{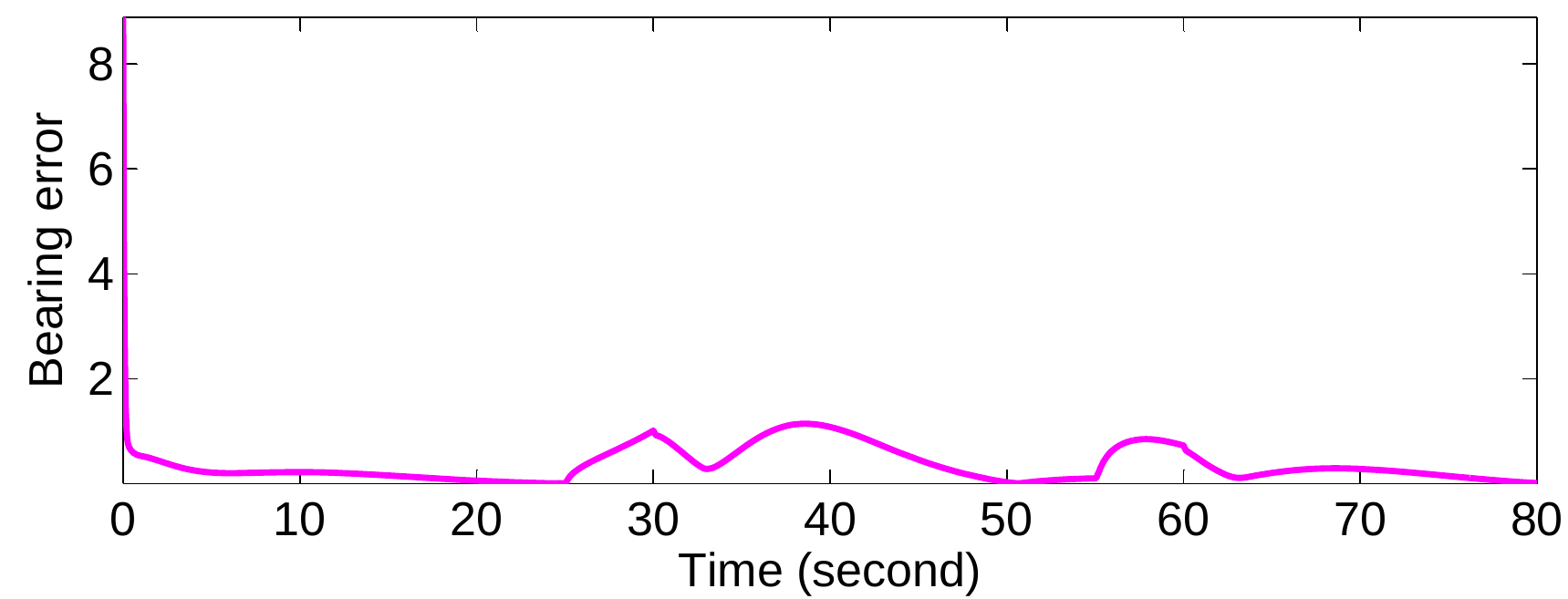}}
  \caption{A 3D example to demonstrate the formation maneuver through a narrow passage.}
  \label{fig_sim_3DObstacle}
\end{figure}

\section{Conclusion}
This paper studied the problem of bearing-based formation maneuver control in arbitrary dimensions.
A distributed PI control law has been proposed and its global stability has been proved.
It was shown that under the proposed control law the leaders can collaborate to control the translation and scale of the formation.
In this paper, we only considered the case where the underlying sensing graph is undirected.
Directed cases will be studied in the future.

{\small
\section*{Acknowledgements}
The work presented here has been supported by the Israel Science Foundation (grant No. 1490/13).
}
\bibliography{myOwnPub,zsyReferenceAll} 
\bibliographystyle{ieeetr}

\end{document}